\documentclass[a4paper]{article}
\usepackage[T1]{fontenc}
\usepackage[letterpaper, left=1in, right=1in, bottom=1in, top=1in]{geometry}
\usepackage{authblk}
\usepackage[pdfencoding=auto]{hyperref}
\usepackage{bookmark}
\usepackage{graphicx}
\usepackage{mathrsfs}
\usepackage{amsmath,amssymb}
\usepackage{amsthm}
\usepackage{algorithmic}
\usepackage{ascmac}
\usepackage{comment}
\usepackage{boites}
\usepackage[ruled,linesnumbered]{algorithm2e}

\newtheorem{theorem}{Theorem}
\newtheorem{corollary}{Corollary}
\newtheorem{lemma}{Lemma}
\newtheorem{observation}{Observation}

\theoremstyle{definition}

\newcommand{\trie}{\mathcal{T}}
\newcommand{\leaves}{\mathsf{leaves}}
\newcommand{\Root}{\mathbf{r}}
\newcommand{\height}{\mathsf{height}}
\newcommand{\rev}[1]{{#1}^{\mathit{R}}}

\newcommand{\lpp}{\mathit{lpp}}

\newcommand{\parent}{\mathsf{parent}}

\newcommand{\str}{\mathsf{str}}
\newcommand{\suf}{\mathsf{suf}}
\newcommand{\slink}{\mathsf{SL}}
\newcommand{\jump}{\mathsf{jump}}

\newcommand{\Substr}{\mathsf{Substr}}
\newcommand{\Suffix}{\mathsf{Suffix}}
\newcommand{\STree}{\mathsf{STree}}
\newcommand{\childchar}{\mathsf{childchar}}
\newcommand{\inSL}{\mathsf{inSLchar}}
\newcommand{\eertree}{\mathsf{eertree}}

\newcommand{\MPal}{\mathsf{MPal}}
\newcommand{\DPal}{\mathsf{DPal}}
\newcommand{\Dest}{\mathsf{destSL}}
\newcommand{\nil}{\mathit{nil}}
\newcommand{\subtree}{\mathsf{subtree}}
\newcommand{\lmDest}{\mathsf{leftDest}}

\begin{document}

\title{
Computing palindromes on a trie in linear time
}
\author[1]{Takuya~Mieno}
\author[2,3]{Mitsuru~Funakoshi}
\author[2,4]{Shunsuke~Inenaga}
\affil[1]{Department of Computer and Network Engineering,\par University of Electro-Communications.\par\texttt{tmieno@uec.ac.jp}}
\affil[2]{Department of Informatics, Kyushu University.\par\texttt{\{mitsuru.funakoshi.inenaga\}@inf.kyushu-u.ac.jp}}
\affil[3]{Japan Society for the Promotion of Science.}
\affil[4]{PRESTO, Japan Science and Technology Agency.}

\date{}

\maketitle

\begin{abstract}
  A trie $\trie$ is a rooted tree
  such that each edge is labeled by a single character from the alphabet,
  and the labels of out-going edges from the same node are mutually distinct.
  Given a trie $\trie$ with $n$ edges, we show how to compute
  all distinct palindromes and all maximal palindromes on $\trie$
  in $O(n)$ time,
  in the case of integer alphabets of size polynomial in $n$.
  This improves the state-of-the-art
  $O(n \log h)$-time algorithms by Funakoshi et al. [PSC 2019],
  where $h$ is the height of $\trie$.
  Using our new algorithms, the eertree
  with suffix links for a given trie $\trie$
  can readily be obtained in $O(n)$ time.
  Further, our trie-based $O(n)$-space data structure allows us to report
  all distinct palindromes and maximal palindromes
  in a query string represented in the trie $\trie$,
  in output optimal time.
  This is an improvement over an existing (na\"ive) solution that precomputes
  and stores all distinct palindromes and maximal palindromes
  for each and every string in the trie $\trie$ separately,
  using a total $O(n^2)$ preprocessing time and space,
  and reports them in output optimal time upon query.
\end{abstract}

\section{Introduction}

A string $p$ is called a \emph{palindrome} if $p$ reads the same forward and backward, namely, $p = \rev{p}$ where $\rev{p}$ denotes the reversed string of $p$. Finding palindromes in a given string is a classical problem in Theoretical Computer Science, with motivations and possible applications in Bioinformatics~\cite{gusfield97:_algor_strin_trees_sequen,WKSung_book}.
A substring palindrome $p = S[i..j]$ in a string $S$ is called a \emph{maximal palindrome} if $S[i-1..j+1]$ is not a palindrome, $i = 1$, or $j = |S|$. Since any substring palindrome in $S$ shares the same center with a unique maximal palindrome, one can essentially obtain all substring palindromes in string $S$ by computing its maximal palindromes. 

There are two well-known algorithms that compute all maximal palindromes in a given string $S$ of length $m$. 
The algorithm of Manacher~\cite{Manacher75} finds all $2m-1$ maximal palindromes in $S$ in $O(m)$ time and space. His algorithm scans an input string from left to right, and works on a general (unordered) alphabet.
The algorithm of Gusfield~\cite{gusfield97:_algor_strin_trees_sequen} consists in two steps: In the preprocessing step, it builds the suffix tree~\cite{Weiner73} of the concatenated string $S\$\rev{S}\#$, where $\$$ and $\#$ are unique symbols not occurring inside $S$, and then enhances the suffix tree with the lowest common ancestor (LCA) data structure~\cite{DBLP:conf/latin/BenderF00} that allows for LCA queries between two nodes in $O(1)$ time after linear-time preprocessing. Then, the algorithm finds all maximal palindromes by applying $2m-1$ \emph{outward} longest common extension (outward LCE) queries in $S$ via the LCA data structure on the suffix tree. Overall, the Gusfield algorithm works in $O(m)$ time and space in the case of \emph{linearly-sortable alphabets}, including constant-size alphabets and integer alphabets of size polynomial in $m$.

While maximal palindromes tell us all \emph{occurrences} of palindromes in a string $S$, the other important direction of research is the \emph{vocabularies} of palindromes in $S$. The latter is called the \emph{distinct palindromes} problem, where the task is to compute the set of all distinct palindromes in $S$. It is known that any string of length $m$ can contain at most $m+1$ distinct palindromes (including the empty string)~\cite{DBLP:journals/tcs/DroubayJP01}.
The algorithm of Groult et al.~\cite{DBLP:journals/ipl/GroultPR10} finds all distinct palindromes in $O(m)$ time and space for linearly-sortable alphabets.

We consider extended versions of the above problems, in which the input is a \emph{trie} (i.e. an edge-labeled rooted tree). Tries are a natural generalization of strings at least in two meanings: A trie can be seen as a generalized string with branches; A trie is a compact representation of a set of strings. 
Funakoshi et al.~\cite{FunakoshiNIBT19} showed that, for any trie $\trie$ with $n$ edges and $l$ leaves, the number of maximal palindromes in $\trie$ is exactly $2n - l$ and the number of distinct palindromes in $\trie$ is at most $n+1$.
In addition, they showed how to find all maximal palindromes in a given trie $\trie$ in $O(n \log h)$ time with $O(h)$ space, where $h$ is the height of $\trie$. This algorithm is based on the periodicity of palindromes, and works on general alphabets.
They also showed how to find all distinct palindromes in a given trie in $O(n \log h)$ time with $O(n)$ space, in the case of linearly-sortable alphabets.

In this paper, we propose the first $O(n)$-time algorithms for computing the maximal palindromes and distinct palindromes in a given trie $\trie$ with $n$ edges, in the case of integer alphabets of polynomial size in $n$.
Our algorithm makes heavy use of the suffix tree of a backward trie~\cite{Kosaraju89a,breslauer98}, but its design is quite different from the aforementioned approach by Gusfield~\cite{gusfield97:_algor_strin_trees_sequen}.
Indeed, no $O(n)$-space $O(1)$-time outward LCE query data structures on tries are known to date (this is because the size of the suffix tree of a forward trie is $\Omega(n^2)$~\cite{Inenaga21}).

Technically speaking, our algorithm is most related to the suffix-tree based offline algorithm of Rubinchik and Shur (Proposition 4.10 in~\cite{RubinchikS18}) that builds the \emph{eertree}, which is a trie-based data structure storing all distinct palindromes in a given string $S$.
In the preprocessing phase of their method, the Manacher algorithm is used to compute all maximal palindromes in $S$. Our first finding in this paper is that this preprocessing phase can indeed be omitted, and instead one can use the suffix links of the leaves to check the maximality of given substring palindromes. This is helpful in our scenario since the application of the Manacher algorithm to a trie takes $O(n \log h)$ time~\cite{FunakoshiNIBT19}. The new suffix-link-based method can simultaneously compute all distinct palindromes (or alternatively the eertree) together with all maximal palindromes, in $O(m)$ time and space, for a given string $S$ of length $m$. However, the time analysis of the suffix-link-based method is due to the fact that each leaf of the suffix tree of a string has exactly one in-coming suffix link (except for the leaf representing the whole string, which has no in-coming suffix link), and the cost of checking an in-coming suffix link to a leaf can be charged to either a unique distinct palindrome or a unique occurrence of a maximal palindrome in $S$. Unfortunately, in our trie case, there can be at most $\sigma$ in-coming suffix links to a single leaf of the suffix tree for a trie, where $\sigma$ is the alphabet size. Since we need to check whether the palindrome in question can be extended with one of at most $\sigma$ candidate characters,
na\"ive suffix-link-based methods would require $\Omega(n \log \sigma)$ time for trie inputs, which is prohibitive for large alphabets. Still, we show an $O(n)$-time suffix-link-based method for computing all distinct/maximal palindromes in the input trie.

We also present how our trie-based $O(n)$-size data structure for storing all maximal/distinct palindromes in the input trie $\trie$ can be used for reporting all maximal/distinct palindromes in the strings stored in $\trie$, in \emph{output optimal} time, without expanding the trie to the plain strings or extracting a string of interest from the trie.
Note that the total length of the strings stored in $\trie$ can be as large as $O(n^2)$.
 \section{Preliminaries}\label{sec:preliminaries}

\subsection{Strings and palindromes}

Let $\Sigma$ be the {\em alphabet}.
An element of $\Sigma^*$ is called a {\em string}.
The length of a string $S$ is denoted by $|S|$.
The empty string $\varepsilon$ is a string of length 0,
namely, $|\varepsilon| = 0$.
For any non-negative integer $k$,
let $\Sigma^k$ denote the set of strings of length $k$.
For a string $S = xyz$, $x$, $y$ and $z$ are called
a \emph{prefix}, \emph{substring}, and \emph{suffix} of $S$, respectively.
A prefix (resp. suffix) of $T$ of length less than $|S|$
is called a \emph{proper prefix} (resp. \emph{proper suffix}) of $S$.

For a string $S$ and an integer $1 \leq i \leq |S|$,
$S[i]$ denotes the $i$th character of $S$,
and for two integers $1 \leq i \leq j \leq |S|$,
$S[i..j]$ denotes the substring of $S$
that begins at position $i$ and ends at position $j$.
For convenience, let $S[i..j] = \varepsilon$ when $i > j$.

Let $\rev{S}$ denote the reversed string of $S$,
i.e., $\rev{S} = S[|S|] \cdots S[1]$.
A string $S$ is called a \emph{palindrome} if $S = \rev{S}$.
We remark that the empty string $\varepsilon$ is also
considered to be a palindrome.
Let $\mathbb{P}$ denote the set of all palindromes over $\Sigma$.

For any substring palindrome $S[i..j]$ in $S$,
$\frac{i+j}{2}$ is called its \emph{center}.
A substring palindrome $S[i..j]$
is said to be a \emph{maximal palindrome} in $S$
if
(a) $S[i-1..j+1]$ is not a palindrome (or equivalently $S[i-1] \neq S[j+1]$),
(b) $i = 1$, or
(c) $j = |S|$.
It is clear that for each $c = 1, 1.5, \ldots, |S|$,
there is a one-to-one correspondence between
$c = \frac{i+j}{2}$ and the maximal palindrome $S[i..j]$. 
Thus, there are exactly $2|S|-1$ maximal palindromes in string $S$.

\subsection{Tries and palindromes}

A {\em trie} $\trie = (V, E)$ is a rooted tree
where each edge in $E$ is labeled by a single character from $\Sigma$
and the out-going edges of each node are labeled by mutually distinct characters.
For any non-root node $\mathbf{u}$ in $\trie$,
let $\parent(\mathbf{u})$ denote the parent of $\mathbf{u}$.
For any node $\mathbf{v}$ in $\trie$,
let $\childchar(\mathbf{v})$ denote the set of out-going edge labels of $\mathbf{v}$.
Let $\height(\trie)$ denote the height of $\trie$.
For convenience, we read the path labels in the input trie $\trie$
in the leaf-to-root direction.
For any path $(\mathbf{u}, \mathbf{v})$,
we denote by $\str(\mathbf{u}, \mathbf{v})$ the path label
from $\mathbf{u}$ to $\mathbf{v}$.
The set $\Substr(\trie)$ of (reversed) \emph{substrings} in $\trie$
is $\Substr(\trie) = \{\str(\mathbf{u}, \mathbf{v}) \mid \mathbf{u},\mathbf{v} \in V, \mbox{$\mathbf{u}$ is a descendant of $\mathbf{v}$}\}$.
Note that $\mathbf{v}$ itself is a descendant of $\mathbf{v}$, i.e., the empty string $\varepsilon = \str(\mathbf{v}, \mathbf{v})$ is an element of $\Substr(\trie)$.
For a string $w$ and a node $\mathbf{u}$ in $\trie$,
the pair $(\mathbf{u}, |w|)$ is called an occurrence of $w$ in $\trie$
if $w$ is a prefix of the substring starting at $\mathbf{u}$ in $\trie$.
Such representations of occurrences allow us
to retrieve the substring $w$ in $O(|w|)$ time
by taking the unique path from $\mathbf{u}$ towards the root.
Also, we enhance the input trie $\trie$ with a \emph{level ancestor} data structure~\cite{bender04:_level_ances_probl} so that any character in a given path $(\mathbf{u}, \mathbf{v})$ can be accessed in $O(1)$ time, after linear-time preprocessing.

We can generalize the notions of maximal palindromes
and distinct palindromes to tries, as follows.
An occurrence of a substring palindrome $p$ which starts at node $\mathbf{u}$ and ends at node $\mathbf{v}$
is called a \emph{maximal palindrome} in $\trie$
if
(a) $\mathbf{u}$ is not a leaf and $\str(\mathbf{u'}, \parent(\mathbf{v}))$ is not a palindrome with \emph{any} child $\mathbf{u'}$ of $\mathbf{u}$,
(b) $\mathbf{u}$ is a leaf, or
(c) $\mathbf{v}$ is the root $\Root$.
Let $\MPal(\trie)$ be the set of all maximal palindromes in $\trie$.
For each $1 \leq k \leq \height(\trie)$,
let $\MPal_k(\trie) = \{(\mathbf{v}, k) \in \MPal(\trie)\}$,
and let $\MPal_0(\trie) = \{\varepsilon\}$.
Namely $\bigcup_{0 \leq k \leq \height(\trie)}\MPal_k(\trie) = \MPal(\trie)$.

For any trie $\trie$,
let $\DPal(\trie)$ be the set of all substring palindromes in $\trie$,
namely $\DPal(\trie) = \Substr(\trie) \cap \mathbb{P}$.
We call the elements of $\DPal(\trie)$ as \emph{distinct palindromes}
in $\trie$.
For each $0 \leq k \leq \height(\trie)$,
let $\DPal_k(\trie) = \DPal(\trie) \cap \Sigma^k$.
Namely $\bigcup_{0 \leq k \leq \height(\trie)}\DPal_k(\trie) = \DPal(\trie)$.
See Figure~\ref{fig:trie} for an example of trie $\trie$ and its substring palindromes.

\begin{theorem}[\cite{FunakoshiNIBT19}] \label{lem:num_pal_trie}
  For any trie $\trie$ with $n$ edges and $l$ leaves,
  $|\MPal(\trie)| = 2n - l$ and $|\DPal(\trie)| \leq n+1$.
\end{theorem}
\begin{figure}[htb]
  \centering
  \includegraphics[width=0.5\linewidth]{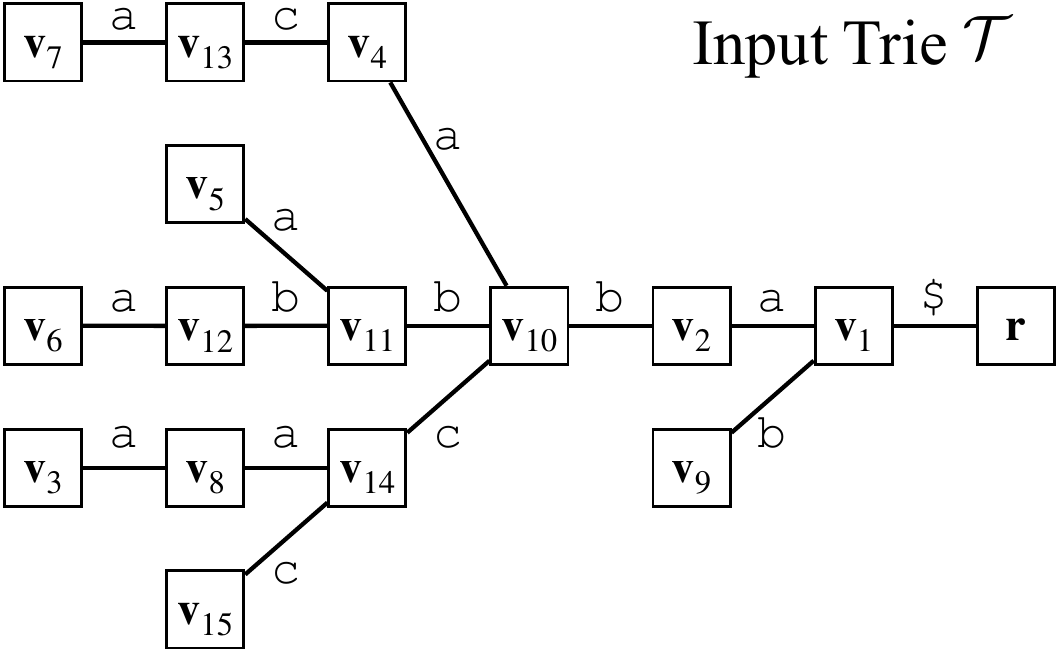}
  \caption{
    An example for an input trie $\trie$ with $n = 15$ edges, in which the path labels are read from the leaves to the root $\mathbf{r}$.
    The height of $\trie$ is $\height(\trie) = 6$.
    Let us focus on the node $\mathbf{v}_{14}$.
    Then $\parent(\mathbf{v}_{14}) = \mathbf{v}_{10}$,
    $\childchar(\mathbf{v}_{14}) = \{\mathtt{a}, \mathtt{c}\}$.
    Also, $\str(\mathbf{v}_{14}, \mathbf{v}_{1}) = \mathtt{cba}$.
    The occurrence $(\mathbf{v}_{10}, 1)$ of palindrome $\mathtt{b}$ is not a maximal palindrome
    since it can be extended to a longer palindrome $\str(\mathbf{v}_4, \mathbf{v}_1) = \mathtt{aba}$.
    This occurrence $(\mathbf{v}_4, 3)$ of $\mathtt{aba}$ is a maximal palindrome.
    We have $\DPal(\trie) = \{\varepsilon, \mathtt{a}, \mathtt{b}, \mathtt{c}, \mathtt{aa}, \mathtt{bb}, \mathtt{cc}, \mathtt{aba}, \mathtt{bbb}, \mathtt{aca}, \mathtt{abba}, \mathtt{abbba}\}$ except for the special character $\$$.
  }
  \label{fig:trie}
\end{figure}

\subsection{Suffix tree and eertree of trie}

\noindent \textbf{\textsf{Suffix tree of a trie}.}
For convenience, we assume that 
the root $\Root$ of the input trie $\trie$
has only a single child and its incoming edge from $\Root$
is labeled by a special character $\$$ that does not occur elsewhere in $\trie$.

For any node $\mathbf{v}$ in $\trie$, let $\suf(\mathbf{v}) = \str(\mathbf{v}, \Root)$.
The set $\Suffix(\trie)$ of (reversed) suffixes of a given trie $\trie = (V, E)$ is $\Suffix(\trie) = \{\suf(\mathbf{v}) \mid \mathbf{v} \in V\}$.

The \emph{suffix tree} of $\trie$, denoted $\STree(\trie)$,
is a compacted trie such that
\begin{itemize}
  \item Each edge of $\STree(\trie)$ is labeled by a non-empty reversed substring of $\trie$;
  \item The labels of the edges from each node of $\STree(\trie)$ begin with mutually distinct characters; 
  \item There is a one-to-one correspondence between the leaves of $\STree(\trie)$ and the non-root nodes in $\trie$ such that every suffix in $\Suffix(\trie)$ is  represented by a unique leaf in $\STree(\trie)$.
\end{itemize}

For any reversed substring $w$ in the input trie $\trie$,
the \emph{locus} for $w$ in the respective suffix tree $\STree(\trie)$
is the ending point of the path which spells out $w$
from the root of $\STree(\trie)$.
Note that the locus is either on a node or in the middle of an edge.
The loci ending in the middle of edges are called
\emph{implicit nodes},
while the loci ending at nodes are called \emph{explicit nodes}.
For any implicit or explicit node $v$ of $\STree(\trie)$,
let $\str(v)$ denote the path string from the root of $\STree(\trie)$ to $v$.
Then, we say that the node $v$ \emph{represents} the substring $\str(v)$.

By the third property of $\STree(\trie)$,
there are exactly $n$ leaves in $\STree(\trie)$.
Since each of the internal explicit nodes has at least two children,
there are at most $n-1$ internal nodes in $\STree(\trie)$.
Thus there are at most $2n-1$ nodes and $2n-2$ edges in $\STree(\trie)$.
Each edge label $x$ of $\STree(\trie)$
is represented by a pair $(\mathbf{u}, \mathbf{v})$ of nodes in $\trie$
such that $x = \str(\mathbf{u}, \mathbf{v})$.
This allows us to represent $\STree(\trie)$ in $O(n)$ space.
See Figure~\ref{fig:suffix_tree} for an example $\STree(\trie)$.

In what follows, we use a common assumption that our alphabet is an integer alphabet of polynomial size in $n$, where 
$n>1$
is the number of edges (and thus the number of characters) in the input trie $\trie$.
We then sort the characters appearing in the input trie $\trie$
in $O(n)$ time by radix-sort,
and replace each character appearing in the input trie $\trie$
with its lexicographical rank.
This ensures that we can work on the integer alphabet $\Sigma = [1..\sigma]$, where $\sigma \leq n$ denotes the number of distinct characters in $\trie$.
Then $\STree(\trie)$ can be built in linear time:

\begin{theorem}[\cite{Shibuya03}] \label{theo:stree_linear_integer}
  For any trie $\trie$ with $n$ edges
  where edge labels are drawn from an integer alphabet $[1..n]$,
  $\STree(\trie)$ can be built in $O(n)$ time and working space.
\end{theorem}

We remark that the algorithm of Theorem~\ref{theo:stree_linear_integer}
builds the \emph{edge-sorted} suffix tree for a given trie $\trie$,
in which the out-going edges of each node are sorted in lexicographical order.
Thus, we can naturally assume that the leaves in $\STree(\trie)$ are arranged in the lexicographical order from left to right.
For each $1 \leq i \leq n$,
let $\ell_i$ denote the $i$th leaf in $\STree(\trie)$
which represents the lexicographically $i$th suffix in $\Suffix(\trie)$.
For each $1 \leq i \leq n$,
let $\mathbf{v}_i$ denote the node in the input trie $\trie$
that corresponds to the leaf $\ell_i$,
namely $\suf(\mathbf{v}_i) = \str(\ell_i)$ (see Figure~\ref{fig:suffix_tree} for examples).

The \emph{suffix link} of any non-root explicit node $v$ in $\STree(\trie)$
points to the explicit node $u$ representing the longest proper suffix of $\str(v)$, i.e. $\str(u) = \str(v)[2..|\str(v)|]$,
and we denote this by $\slink(v) = u$.
Note that the destination explicit node $u$ always exists.
Also, by the definition, for each leaf $\ell_i$ in the suffix tree,
  $\slink(\ell_i)=\ell_j$ if and only if $\parent(\mathbf{v}_i) = \mathbf{v}_j$.
The first character $\str(v)[1]$ that is dropped when moving from $v$ to $u$
is called the \emph{suffix link label} of $\slink(v)$.
\begin{figure}[htb]
  \centering
  \includegraphics[width=0.7\linewidth]{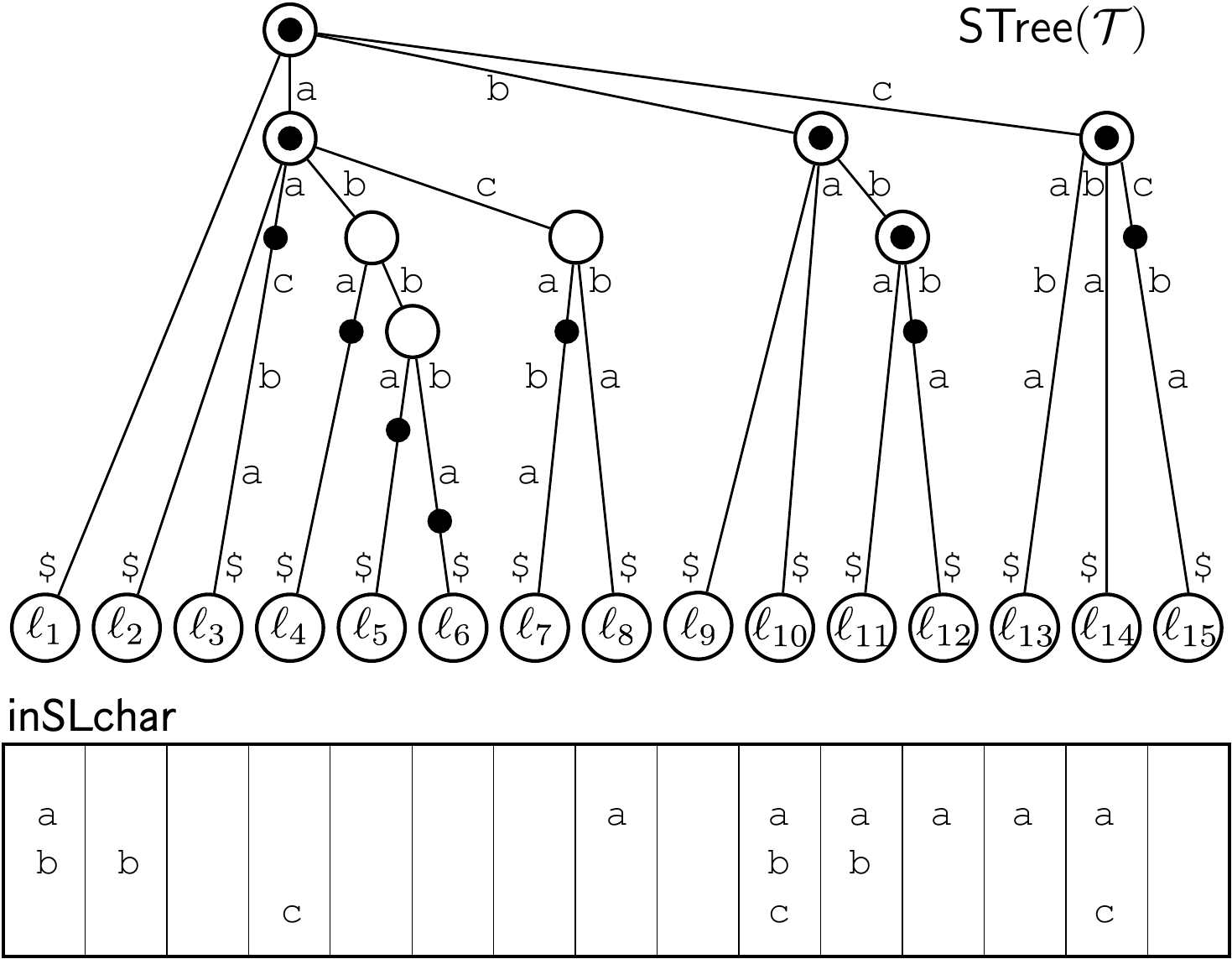}
  \caption{
    The suffix tree $\STree(\trie)$ of the input trie $\trie$ of Figure~\ref{fig:trie}.
    The explicit nodes are depicted by the white circles,
    while the palindromic nodes representing the elements of $\DPal(\trie)$ are depicted by the black circles.
    The sets $\inSL(\ell_i)$ for all $\ell_i$ are depicted below the leaves.
    Each $\inSL(\ell_i)$ stores the in-coming suffix link labels of leaf $\ell_i$.
    The suffix links are omitted in this figure, however, we can verify this $\inSL$ is correct by looking at trie $\trie$
    since $\inSL(\ell_i) = \childchar(\mathbf{v}_i)$ holds.
  }
  \label{fig:suffix_tree}
\end{figure}

In our algorithms to follow, we will use only
the suffix links of the leaves of $\STree(\trie)$.
Since every suffix of $\trie$ ends with the unique end-marker $\$$,
the suffix link of a leaf points to another leaf,
except for the suffix link of the leaf representing $\$$
which points to the root.
For any leaf $\ell_i$ in $\STree(\trie)$,
let $\inSL(\ell_i)$ denote the set of in-coming suffix link labels.
Notice that $\inSL(\ell_i) = \childchar(\mathbf{v}_i)$ holds
since $\slink(\ell_j)=\ell_i$ iff $\parent(\mathbf{v}_j) = \mathbf{v}_i$.
See Figure~\ref{fig:suffix_tree} for examples.
For any node $v$ in $\STree(\trie)$,
let $\subtree(v)$ denote the subtree rooted at $v$, and
let $\leaves(v)$ denote the set of leaves in $\subtree(v)$.
Given $\STree(\trie)$,
it is easy to compute the suffix links of all leaves of $\STree(\trie)$
in $O(n)$ time.
Thus we can also compute $\inSL(\ell)$ for every leaf $\ell$ in $\STree(\trie)$ in a total of $O(n)$ time.

We extend the notion of $\childchar(\cdot)$ to the suffix tree so that for each node $v$ in $\STree(\trie)$, $\childchar(v)$ denotes the set of the first characters of the out-going edge labels of $v$.

\vspace*{0.5pc}
\noindent \textbf{\textsf{Eertree of a trie}.}
We can naturally generalize the \emph{eertree} (a.k.a. \emph{palindromic tree})
which represents the distinct palindromes of a single string~\cite{RubinchikS18} to
the case of a trie $\trie$, as follows:
The eertree for a trie $\trie$, denoted $\eertree(\trie)$,
is a pair of two rooted tries (the even-tree and odd-tree)
where the root of the even-tree represents $\varepsilon$
and the root of the odd-tree is an auxiliary node $\bot$.
Then, for each $p \in \DPal(\trie)$,
there is a path that spells out $p[\lfloor |p|/2 \rfloor+1.. |p|]$
from the root of the even-tree if $|p|$ is even,
or from the root of the odd-tree otherwise.
The \emph{suffix link} of a node representing a palindrome $p$ in $\eertree(\trie)$
points to the node that represents the longest palindrome $q$, which is a proper suffix of $p$ (the suffix links may bridge the nodes of the even-tree and the nodes of the odd-tree).
The suffix links of the two roots $\varepsilon$ and $\bot$
both point to $\bot$.
See Figure~\ref{fig:eertree}
for an example of the eertree for a trie.
\begin{figure}[htb]
  \centering
  \includegraphics[width=0.5\linewidth]{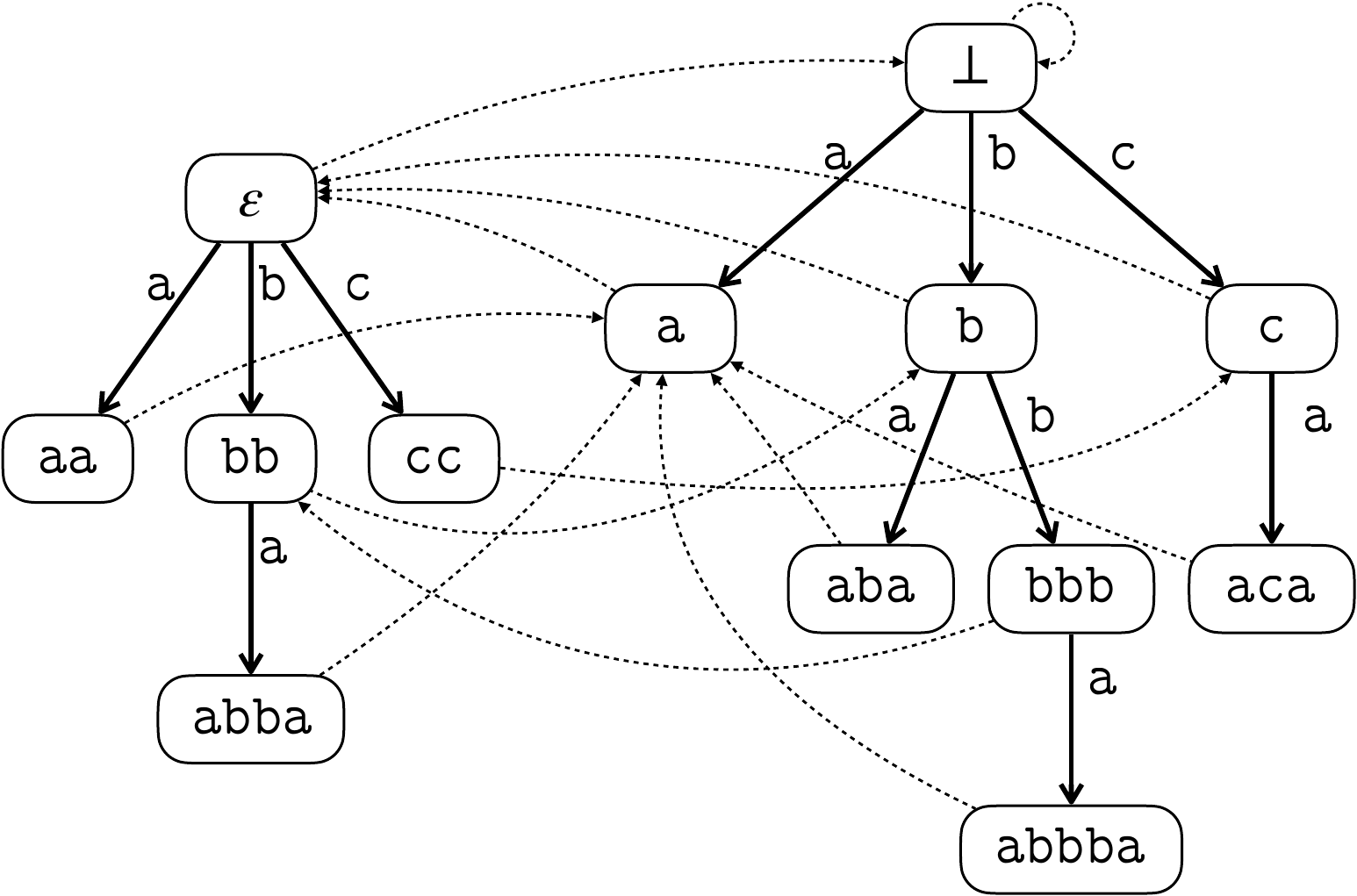}
  \caption{
    The palindromic tree $\eertree(\trie)$ for the input trie $\trie$ of Figure~\ref{fig:trie}, which represents
    $\DPal(\trie) = \{\varepsilon, \mathtt{a}, \mathtt{b}, \mathtt{c}, \mathtt{aa}, \mathtt{bb}, \mathtt{cc}, \mathtt{aba}, \mathtt{bbb}, \mathtt{aca}, \mathtt{abba}, \mathtt{abbba}\}$.
    The bold arcs represent the edges and the dotted arcs represent the suffix links.
  }
  \label{fig:eertree}
\end{figure}
 \section{Algorithm for computing \texorpdfstring{$\DPal(\trie)$}{DPal(T)} and \texorpdfstring{$\MPal(\trie)$}{MPal(T)}} \label{sec:algo}

This section presents our algorithm for computing
$\MPal(\trie)$ and $\DPal(\trie)$ for a given trie $\trie$.

\subsection{Overview of our algorithm}

We use $\STree(\trie)$ as our main tool for computing $\DPal(\trie)$ and $\MPal(\trie)$.
An (implicit or explicit) node $v$ on $\STree(\trie)$
such that $\str(v) \in \mathbb{P}$ is called a \emph{palindromic node}.
See also Figure~\ref{fig:suffix_tree}, in which the palindromic nodes
are depicted by the black circles.
By Theorem~\ref{lem:num_pal_trie},
there are exactly $|\DPal(\trie)| \in O(n)$ palindromic nodes in $\STree(\trie)$,
and thus, computing $\DPal(\trie)$ can be reduced to computing all palindromic nodes in $\STree(\trie)$.

We compute $\DPal(\trie)$ together with $\MPal(\trie)$ 
in increasing order of the lengths of the palindromes.
Namely, we visit all palindromic nodes in $\STree(\trie)$,
in the level-wise breadth first manner.
Also, for each found palindrome $p$ in increasing order of their length,
we check whether it has an occurrence as a maximal palindrome in $\trie$.
This can be done as follows:

\begin{itembox}[c]{Our strategy for computing $\DPal(\trie)$ and $\MPal(\trie)$}
\begin{itemize}

  \item \textbf{Base steps:} Start from the root $\Root$ of $\STree(\trie)$. Report $\varepsilon$ as a member of $\DPal_0(\trie)$, and report all single characters occurring in $\trie$, as members of $\DPal_1(\trie)$.

\item \textbf{Inductive step:} For each increasing $k = 0, 1, \ldots, \height(\trie)-1$:
 \begin{itemize}
  \item For each palindrome $p \in \DPal_k(\trie)$, let $v$ be the locus for $p$ in $\STree(\trie)$.
    \begin{itemize}
      \item For each character $c \in \childchar(v)$, perform the following:
      \begin{itemize}
\item If $cpc \in \Substr(\trie)$, then report this new palindrome $cpc$ of length $k+2$ as a new member of $\DPal_{k+2}(\trie)$.

      \item For each occurrence of $pc$ that is not preceded by $c$ in the reversed trie $\trie$, report this occurrence of $p$ as a new member of $\MPal_k(\trie)$.
\end{itemize}
    \end{itemize}
 \end{itemize}
 \end{itemize}
\end{itembox}

Given a palindrome $p$ of length $k$ together with the corresponding palindromic node $v$ in $\STree(\trie)$ and $c \in \childchar(v)$,
our task is to efficiently check whether or not $cpc \in \Substr(\trie)$,
and if so, then find the locus for $cpc$ in $\STree(\trie)$. 
In the sequel, for simplicity,
we identify the strings $p$ and $pc$
with their loci (implicit or explicit nodes) in $\STree(\trie)$.

For a palindrome $p$ and a character $c$,
$cpc \in \Substr(\trie)$ iff there exists a node $\mathbf{v}_i \in \trie$ such that $\suf(\mathbf{v}_i)[1..|p|+1] = pc$ and
$\mathbf{v}_i$ has a child that is reachable with character $c$.
Recalling that this is equivalent to that $c \in \inSL(\ell_i)$ with the suffix tree leaf $\ell_i$ that corresponds to $\mathbf{v}_i$, we obtain the two following observations (see also Figure~\ref{fig:obs}).

\begin{observation} \label{obs:distinct_pals_on_stree}
  Let $p$ be a locus in $\STree(\trie)$ such that $p \in \DPal_k(\trie)$,
  and let $c \in \childchar(p)$.
  Then, $cpc \in \DPal_{k+2}(\trie)$ 
  iff there exists a suffix tree leaf $\ell \in \leaves(pc)$ such that
  $c \in \inSL(\ell)$.
\end{observation}

\begin{observation} \label{obs:maximal_pals_on_stree}
  Let $p$ be a locus in $\STree(\trie)$
  such that $p \in \DPal_k(\trie)$,
  and let $c \in \childchar(p)$.
  Then, $(\mathbf{v}_i, |p|) \in \MPal_k(\trie)$
  iff $\ell_i \in \leaves(pc)$ and $c \notin \inSL(\ell_i)$,
  where $\ell_i$ is the suffix tree leaf corresponding to $\mathbf{v}_i$.\end{observation}
\begin{figure}[ht]
  \centering
  \includegraphics[width=0.7\linewidth]{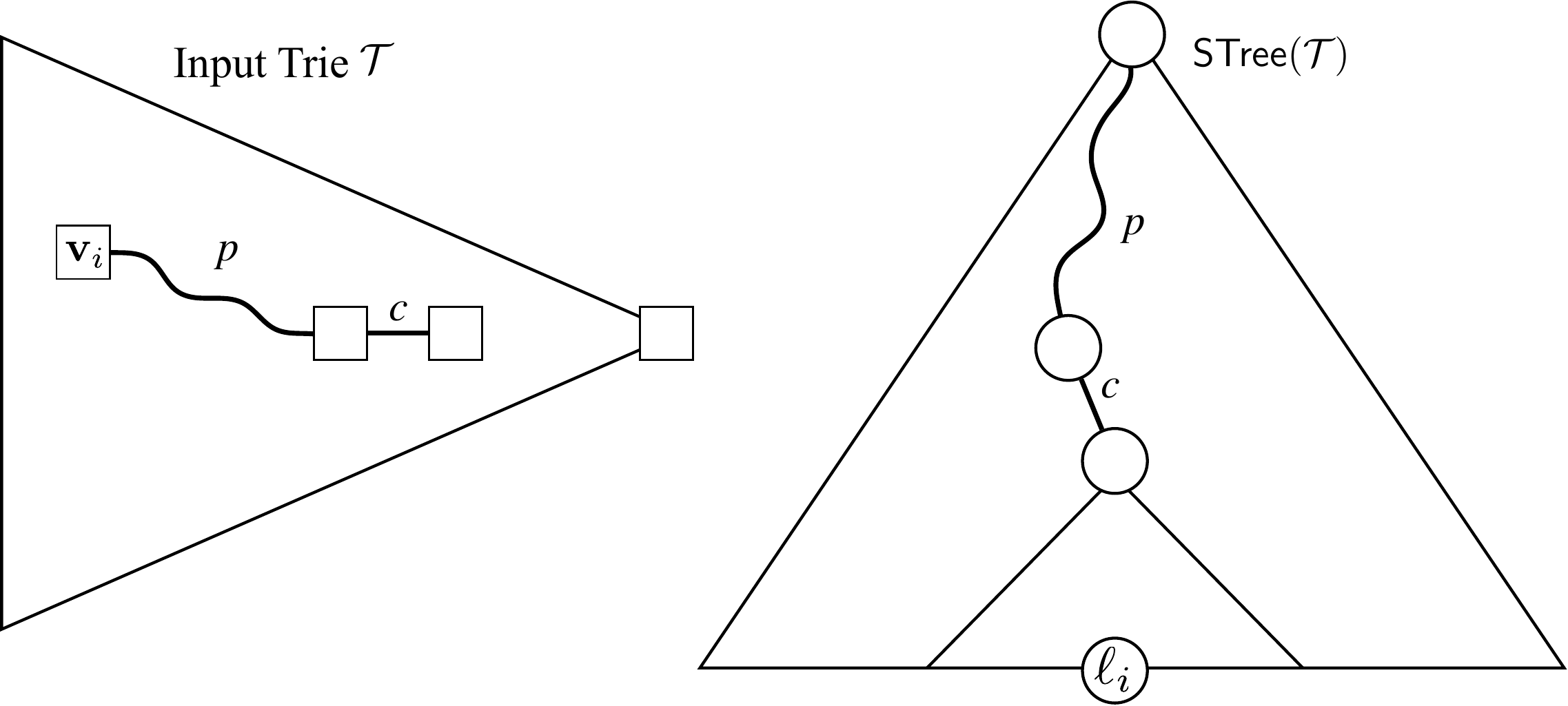}
  \caption{
    Illustration for Observations~\ref{obs:distinct_pals_on_stree} and \ref{obs:maximal_pals_on_stree}.
  } \label{fig:obs}
\end{figure}  

Following Observation~\ref{obs:distinct_pals_on_stree} and Observation~\ref{obs:maximal_pals_on_stree},
our task is, given a palindromic node $p$ in $\STree(\trie)$
and $c \in \childchar(p)$, to perform the following task in \emph{(amortized) constant time}.
\begin{enumerate}
\item[(1)] Check whether there exists a leaf $\ell$ under the locus for $pc$ such that $\inSL(\ell)$ contains character $c$, and if so, find it.
  Then, locate the locus for the extended palindrome $cpc$ in $\STree(\trie)$.
  
\item[(2)] Check whether there exists a leaf $\ell$ under the locus for $pc$ such that $\inSL(\ell)$ does not contain character $c$, and if so, retrieve all and only such leaves by skipping the other leaves which do not correspond to the output with respect to $pc$.
\end{enumerate}

In what follows, we show the details of our algorithm that achieves the above goal,
which consists of the preprocessing phase and the computing phase.

\subsection{Preprocessing phase}

We first enhance $\STree(\trie)$ with a lowest common ancestor (LCA) data structure in linear preprocessing time so that the LCA of two given nodes can be found in $O(1)$ time~\cite{DBLP:conf/latin/BenderF00}.
We then compute the \emph{jump-arrays} and some links toward leaves, as follows.

\vspace*{1pc}
\noindent $\Dest$ \textbf{\textsf{arrays and jump-arrays.}}
For each character $c$ occurring in $\trie$,
let $\Dest_c$ be an array such that
$\Dest_c[r] = i$ if the leaf $\ell_i$ is the lexicographically $r$-th leaf that has an in-coming suffix link labeled $c$.
For convenience, we add an extra element at the end of $\Dest_c$ that stores $\infty$.
Also, we sometimes regard $\Dest_c$ as a list of the corresponding leaves $\ell_i$.
We then define the \emph{jump-array} $\jump_c$ which leads us to
the ending point of the run of adjacent leaves in $\STree(\trie)$
for each given leaf $\ell_i$ in $\Dest_c$.
Formally, $\jump_c[i] = i+1$ if $\Dest_c[i+1]-\Dest_c[i] > 1$, and
$\jump_c[i] = \min\{ j \mid j > i \text{ and } \Dest_c[j+1]-\Dest_c[j] > 1\}$ otherwise.
See Figure~\ref{fig:suffix_tree_two} for examples.
Note that the total sizes of $\Dest_c$ and $\jump_c$
for all $c$ occurring in $\trie$ is $O(n)$
since these are linear in the number of suffix links of the leaves.
Also, we can compute the arrays $\Dest_c$ and $\jump_c$
for all characters $c$ occurring in $\trie$ in linear total time,
by radix-sorting all
the pairs $(c, i)$ of a character and a leaf-index such that $c \in \inSL(\ell_i)$.
Simultaneously, we associate each pair $(c, i)$ with
its \emph{rank} in $\Dest_c$, i.e., the integer $r$ with $\Dest_c[r] = i$.

\vspace*{1pc}
\noindent \textbf{\textsf{Links toward leaves.}} 
For each non-leaf node $v$ in $\STree(\trie)$,
we store links from $v$ to the leftmost and the rightmost leaves in $\leaves(v)$.
Note that we can compute them easily by performing a depth-first traversal on $\STree(\trie)$.
We further store another link defined below in every non-root node $v$.
For each non-root node $v$ in $\STree(\trie)$,
let $c_v$ be the first character of its in-coming edge label,
and further let $\lmDest(v)$ be (some representation of)
the leftmost leaf in $\leaves(v)$ such that its leaf-index is in $\Dest_{c_v}$ if such a leaf exists, and $\nil$ otherwise.
For technical reasons, we implement $\lmDest(v)$ as the index of the entry of $\Dest_{c_v}$ corresponding to the specified leaf.
See Figure~\ref{fig:suffix_tree_two} for examples.
Whenever we access an entry of $\Dest_c$, we always take this link from a given node $v$ such that $c_v = c$.
Our implementation of $\lmDest(v)$ ensures that we can access
the corresponding entry in $O(1)$ time, without the need to
search the array $\inSL(\ell)$ on some leaf $\ell \in \leaves(v)$ for character $c_v$, which would use $O(\log \sigma)$ time.
Next, we show how to compute the $\lmDest$ for all nodes in linear total time.
\begin{figure}[htb]
  \centering
  \includegraphics[width=0.7\linewidth]{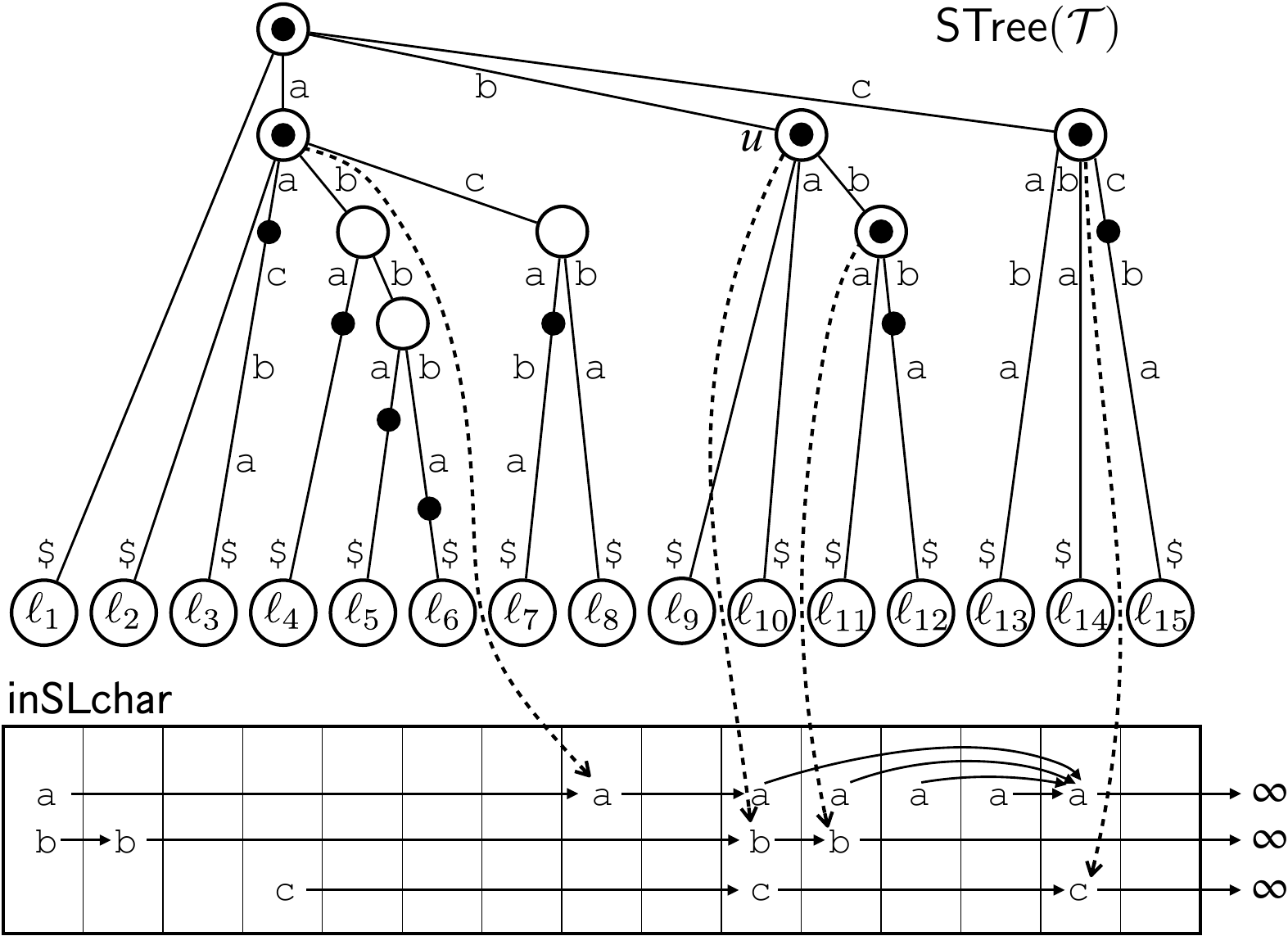}
  \caption{
    The suffix tree $\STree(\trie)$ of Figure~\ref{fig:suffix_tree} with some auxiliary data structures.
By reading array $\inSL$ horizontally for each character $c$, we can obtain the array $\Dest_c$.
Here, $\Dest_{\mathtt{a}} = (1,8,19,10,11,12,13,14,\infty)$,
    $\Dest_{\mathtt{b}} = (1,2,10,11,\infty)$, and
    $\Dest_{\mathtt{c}} = (4,10,14,\infty)$.
Each element in $\Dest_c$ is represented by the character $c$ and is aligned along the corresponding leaf.
    Also, the arc from each element of $\Dest_c$ represents the value stored in the corresponding entry of the $\jump_c$ array.
    Each dotted arrow from an internal node $v$ denotes $\lmDest(v)$.
    For example, $\lmDest(u)$ of the node $u$ representing string $\mathtt{b}$ points to
    the third entry of $\Dest_{\mathtt{b}}$,
    which is represented by the $\mathtt{b}$ under the leaf $\ell_{10}$,
    since the first character $c_u$ of the label of $u$'s parent edge is $c_u = \mathtt{b}$,
    and the leaf $\ell_{10}$ is
    the leftmost leaf in $\leaves(u)$ such that its leaf-index is in $\Dest_\mathtt{b}$.
  }
  \label{fig:suffix_tree_two}
\end{figure}

As a warm-up, we fix a character $c$ and describe an algorithm
which computes $\lmDest(v)$ for all nodes $v$ such that $c_v = c$, in $O(n)$ time.
We execute left-to-right depth-first traversal on $\STree(\trie)$
with a stack $\mathit{st}_c$ storing nodes.
When we descend to a node $v$, push $v$ to $\mathit{st}_c$ if $c_v = c$.
When we ascend from a node $v$, pop $v$
and set $\lmDest(v) \leftarrow \nil$
if the top of $\mathit{st}_c$ is $v$.
When we arrive at a leaf $\ell_k$, scan $\inSL(\ell_k)$ and if $c \in \inSL(\ell_k)$~(i.e., $k \in \Dest_c$),
pop all nodes $v$ in $\mathit{st}_c$,
and set $\lmDest(v) \leftarrow r$ with $\Dest_c[r] = k$.
Recall that such index $r$ has been associated with $(c, k)$ for $c \in \inSL(\ell_k)$.
During the traversal, $\mathit{st}_c$ stores all nodes $v$ on the path from the root to the current node
s.t. $c_v = c$ but $\lmDest(v)$ is not determined yet.
Thus, when arriving at a leaf $\ell_k$ with $c \in \inSL(\ell_k)$,
we can correctly compute $\lmDest(v)$ for all $v$ s.t. $\Dest_c[\lmDest(v)] = k$.

Now, let us remove the constraint of fixing the character $c$,
and perform linear-time computation of $\lmDest(v)$ for all internal nodes $v$.
To compute $\lmDest(v)$ for \emph{all} non-root nodes $v$,
we run the above algorithm for all $c \in \Sigma$ simultaneously in a single depth-first traversal.
For this, we use a length-$n$ array $\mathbf{S}$ of stacks
in which each entry $\mathbf{S}[c]$ plays the role of $\mathit{st}_c$ as the above.
While performing a left-to-right depth-first traversal,
when we descend to a node $v$, push $v$ to $\mathbf{S}[c_v]$,
and when we ascend from a node $v$, pop $v$
and set $\lmDest(v) \leftarrow \nil$
if the top of $\mathbf{S}[c_v]$ is $v$.
When we arrive at a leaf $\ell_k$, then scan $\inSL(\ell_k)$ and for every $c \in \inSL(\ell_k)$, pop all nodes $v$ in $\mathbf{S}[c]$ and set $\lmDest(v) \leftarrow r$ with $\Dest_c[r] = k$.
The correctness relies on the aforementioned algorithm, and the running time is linear in the size of $\STree(\trie)$ including the suffix links,
that is, $O(n)$.

 \subsection{Computing phase}
In the computing phase, we traverse on $\STree(\trie)$
and visit the loci representing palindromes in non-decreasing order of the lengths of palindromes.
Let us assume that we are now located at a node representing a palindrome $p$ occurring in $\trie$.
There are the two following cases: (A) node $p$ is an explicit node, and (B) node $p$ is an implicit node.

\paragraph*{Algorithm for case (A)}
For the case (A),
we can enumerate all trie nodes whose length-$|p|$ prefixes are maximal palindromes $p$
in output-sensitive time by using data structures precomputed.
Let $v$ be a child of node $p$ and $c$ be the first character of the label of edge $(p, v)$.
Let $\ell_{\mathit{left}}$ and $\ell_{\mathit{right}}$ be the leftmost and the rightmost leaves in $\leaves(v)$, respectively.
Namely, they correspond to the $\mathit{left}$-th and $\mathit{right}$-th smallest suffixes in $\trie$, respectively.
We initialize the current-leaf-index $\mathit{curr} \leftarrow \Dest_{c_v}[\lmDest(v)]$ and
the previous-leaf-index $\mathit{prev} \leftarrow \mathit{left}-1$.
While $\mathit{curr} \le \mathit{right}$,
output trie nodes $\mathbf{v}_i$ for all $i$ with $\mathit{prev} < i < \mathit{curr}$ if ${\mathit{curr}-1} \not\in\Dest_c$,
and then set $\mathit{prev} \leftarrow \mathit{curr}$
and $\mathit{curr} \leftarrow \Dest_c[\jump_c[\mathit{curr}]]$.
At last, after breaking out from the while-loop condition, output trie nodes $\mathbf{v}_i$ for all $i$ with $\mathit{prev} < i \le \mathit{right}$
if ${\mathit{curr}-1} \not\in\Dest_c$.
See also Figure~\ref{fig:jumpalgo} for illustration.
\begin{figure}[ht]
  \centering
  \includegraphics[width=0.8\linewidth]{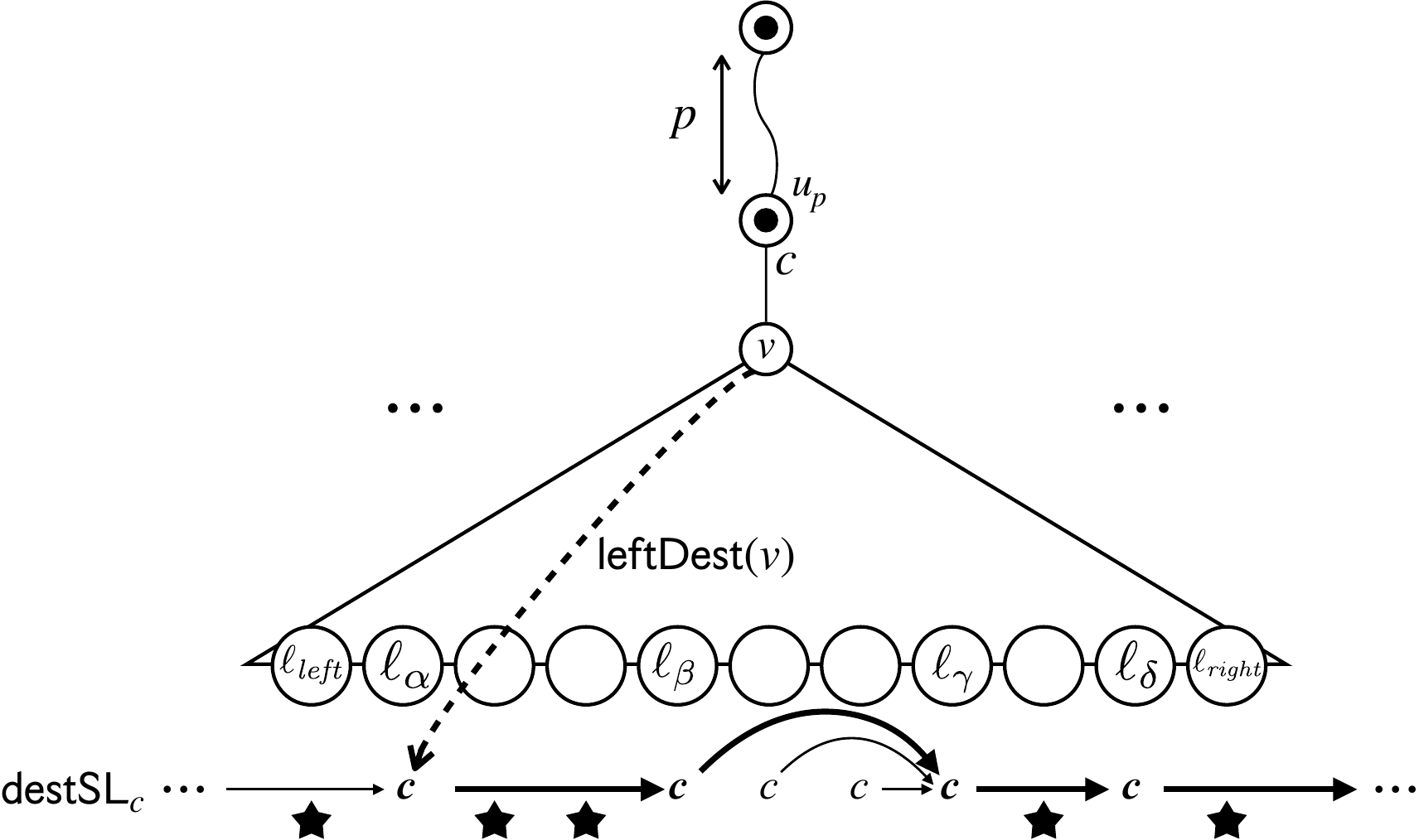}
  \caption{
    The black stars indicate the leaves in $\leaves(v)$ corresponding to occurrences of palindrome $p$
    which are followed by character $c$ but cannot be expanded with $c$, i.e., occurrences of $p$ as maximal palindromes.
    In the computing phase, starting from node $v$,
    we visit the entries in $\Dest_c$ corresponding to the leaves $\ell_\alpha$, $\ell_\beta$, $\ell_\gamma$ and $\ell_\delta$
    in this order,
    and output all the leaves indicated by black stats.
    The jump-arrays values that are used in this step are depicted in bold.
    Note that when we arrive at $\ell_\gamma$, we output nothing
    since none of the leaves between $\ell_\beta$ and $\ell_\gamma$ inclusive corresponds to an occurrence of a maximal palindrome.
  } \label{fig:jumpalgo}
\end{figure}By Observation~\ref{obs:maximal_pals_on_stree}, this algorithm correctly reports all the occurrences of $pc$
(without duplication)
such that the occurrence of $p$ is a maximal palindrome.
Next, to evaluate the running time, we show the next lemma:
\begin{lemma} \label{lem:num_of_jump}
  In the above algorithm for a node $v$,
  when we use jump-pointers three times, at least one maximal palindrome is reported.
\end{lemma}
\begin{proof}
  We assume that we use jump-pointers at least three times on leaves in $\subtree(v)$.
  Then, at least two jumps are performed inside $\subtree(v)$.
  (The last jump, which is toward the outside of $\subtree(v)$, is the only exception.)

  Let $i$ be the current leaf-index.
  By the definition of jump-pointers, if $\Dest_c[i+1] - \Dest_c[i] > 1$, then $\jump_c[i] = i+1$.
  Then, our algorithm reports at least one maximal palindrome
  since we skip at least one leaf between $\Dest_c[i]$ and $\Dest_c[\jump_c[i]] = \Dest_c[i+1]$ both exclusive.
Otherwise, namely if $\Dest_c[i+1] - \Dest_c[i] = 1$,
  then $\jump_c[i]$ points to the ending point of the run of adjacent leaves in $\Dest_c$.
  In this case, we report no maximal palindromes.
  However, in the next step, we track $\jump_c[j]$ where $j = \jump_c[i]$.
  Then, we fall into the first case $\Dest_c[j+1] - \Dest_c[j] > 1$ since $j$ is the ending position of a run
  (see also Figure~\ref{fig:jumpalgo}).
  Thus, a maximal palindrome is reported while the second jump.
Therefore, performing jumps twice inside $\subtree(v)$ results in a maximal palindrome being reported.
\end{proof}
Thus, the running time of the above algorithm for node $v$ is linear in the output size,
or is constant if there is no maximal palindrome to output.

After these processes for $v$, we retrieve the node representing palindrome $cpc$ if it exists.
While computing maximal palindromes we can obtain
the rightmost leaf $\ell_R$ such that $\ell_R\in\leaves(v)$ and $R\in\Dest_c$.
We move to the source leaves $\ell$ and $\ell'$ of the suffix links which respectively point to
$\ell_L$ and to $\ell_R$, both labeled by $c$
where $L = {\Dest_{c_v}[\lmDest(v)]}$ is the leaf-index specified by $\lmDest(v)$.
We then compute the LCA $u'$ of $\ell$ and $\ell'$.
If the string-depth of $u'$ equals $|cpc|$, the node $u'$ is the desired node corresponding to $cpc$.
Otherwise, the locus on the parent-edge of $u'$ with the string-depth of $|cpc|$ is the (implicit) node corresponding to $cpc$.

\paragraph*{Algorithm for case (B)}
For the case (B), the precomputed data structures are not helpful.
Instead, we employ the following lemmas for designing an efficient algorithm.

\begin{lemma}\label{lem:implicit_pal_node}
  For any trie $\trie$,
  there is at most one implicit node $v$ on each edge of $\STree(\trie)$
  such that $\str(v)$ is a palindrome.
\end{lemma}

\begin{proof}
  Suppose on the contrary that there are two implicit nodes $v$ and $u$
  on the same edge of $\STree(\trie)$
  such that both $\str(v)$ and $\str(u)$ are palindromes.
  Assume without loss of generality that $|\str(v)| < |\str(u)|$.
  Since $v$ and $u$ are on the same edge,
  we have $\leaves(v)= \leaves(u)$ which means that
  the number of occurrences of $\str(v)$ and $\str(u)$ in $\trie$ must be equal.
  However, since $\str(v)$ is a proper prefix of $\str(u)$
  and both of them are palindromes,
  $\str(v)$ is also a proper suffix of $\str(u)$,
  meaning that $\str(v)$ occurs at least twice in $\str(u)$.
This contradicts that the number of occurrences of $\str(v)$ and $\str(u)$ in $\trie$ are equal.
Thus there exists at most one implicit node on each edge of $\STree(\trie)$.
\end{proof}

The following lemma states that implicit palindromic nodes have a sort of monotonicity.
See also Figure~\ref{fig:implicit_descendant} for illustration.
\begin{lemma}\label{lem:monotonicity_of_implicit_pal}
  Let $p$ and $q$ be palindromes such that $p$ is a proper prefix of $q$.
  If $p$ corresponds to an implicit node on an edge $(u, v)$ in $\STree(\trie)$ and its following character is $c$,
  then $q$ corresponds to an implicit node on an edge in $\subtree(v)$ and its following character is also $c$.
\end{lemma}
\begin{proof}
Since $p$ corresponds to an implicit node, the only character following $p$ is $c$ for all occurrences of $p$ in $\trie$.
  Also, since the palindrome $p$ is a proper prefix of the palindrome $q$, $p$ occurs in $q$ also as a proper suffix.
  If we assume that $q$ corresponds to an explicit node, then there have to be two distinct characters following $q$, and $p$ as well, a contradiction.
Thus, $q$ corresponds to an implicit node and its following character is $c$.
Also, by Lemma~\ref{lem:implicit_pal_node}, the implicit node cannot exist on the edge $(u, v)$, and thus, it is in $\subtree(v)$.
\end{proof}
Namely, once an implicit node $u$ corresponding to a palindrome is found,
any palindrome corresponding to a descendant of $u$ is guaranteed to be followed by the same character (see also Figure~\ref{fig:implicit_descendant}).
Based on this monotonicity, we design an algorithm for the case (B).
\begin{figure}[htb]
  \centering
  \includegraphics[width=0.5\linewidth]{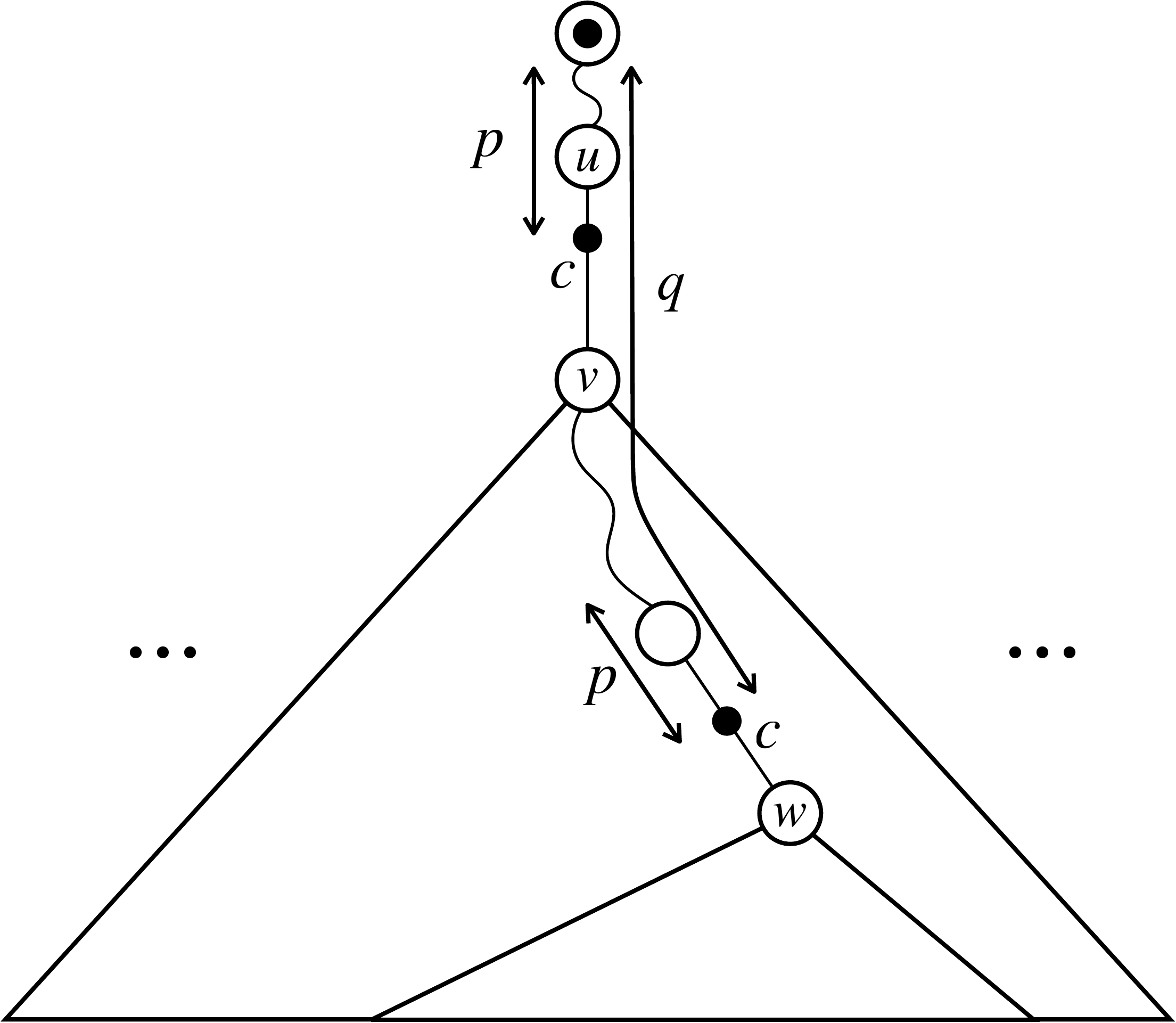}
  \caption{
Illustration for the case (B).
    In our algorithm, we visit the locus of palindrome $p$ before that of palindrome $q$ since $q$ is longer than $p$.
    After we have visited the locus of $p$, each node $w'$ in $\subtree(v)$ have links $\mathsf{left}(w')$ w.r.t. the character $c$.
    When we visit the locus of $q$ later, we use $\mathsf{left}(w)$ and $\Dest_{c}$ to report all occurrences of $q$ which are maximal palindromes.
}
  \label{fig:implicit_descendant}
\end{figure}

At each step of the algorithm, $v.\mathsf{mark} = \mathtt{1}$ means that we already visited an implicit ancestor representing some shorter prefix palindrome.
Initially, we set $u.\mathsf{mark} \leftarrow \mathtt{0}$ for every node $u$.
Let $v$ be the shallowest explicit node below $p$ and $c$ be the character following $p$.
If $v.\mathsf{mark} = \mathtt{0}$, every $v^\prime \in \subtree(v)$ is not marked yet
since we consider palindromes in non-decreasing order of the lengths.
We traverse $\subtree(v)$ and mark all nodes in the subtree.
Also, while traversing the subtree,
we compute the links $\mathsf{left}(v')$ from each $v' \in \subtree(v')$ defined as follows:
$\mathsf{left}(v')$ is the leftmost leaf in $\leaves(v')$ such that
its leaf-index is in $\Dest_{c}$ if such a leaf exists, and $\nil$ otherwise.
In other words, $\mathsf{left}(v')$ is a variant of $\lmDest(v')$ where the character $c$ to be considered is fixed.
Traversing $\subtree(v)$ and
computing $\mathsf{left}(v')$ for all nodes $v'$ in $\subtree(v)$ can be done
in time linear in the size of $\subtree(v)$
plus the total number of incoming suffix links of $\leaves(v)$
as in the computing of $\lmDest$ described in the previous subsection.
Once we create the links $\mathsf{left}$,
we can compute all occurrences of maximal palindrome $p$ and find the next node for $cpc$ in a similar way to the case (A).

If $v.\mathsf{mark} = \mathtt{1}$, there is an implicit ancestor we already visited.
By the above procedures, node $v$ already has link $\mathsf{left}(v)$,
and hence, we can output all maximal palindromes $p$ in output linear time.
To summarize, the following theorem holds:
\begin{theorem}
  For a given trie $\trie$ with $n$ edges over an integer alphabet $\Sigma = [1..n]$,
  $\MPal(\trie)$ and $\DPal(\trie)$ can be computed in $O(n)$ time and space.
\end{theorem}
\begin{proof}
  In the preprocessing phase, $\STree(\trie)$ can be constructed in $O(n)$ time and space (Theorem~\ref{theo:stree_linear_integer}).
  An LCA data structure on $\STree(\trie)$ can be built in $O(n)$ time and space~\cite{DBLP:conf/latin/BenderF00},
  and the remaining data structures can also be built in $O(n)$ time
  by radix-sort and depth-first traversal on $\STree(\trie)$.
  
  For the case (A) in the computing phase,
  most of the procedures can be done in constant time per an edge of $\STree(\trie)$.
The exception is the number of times jump-arrays are used can not be bounded by a constant per an edge.
However,
by Lemma~\ref{lem:num_of_jump},
that can be charged to the maximal palindromes to output, and thus, the total execution time is $O(n+|\MPal(\trie)|) = O(n)$.
For the case (B) in the computing phase, traversing the subtree of a node appears to be the bottleneck.
  However, the total time to traverse the subtrees is $O(n)$ since each node will be marked at most once throughout the algorithm.
  The remaining procedures can be done in a total of $O(n)$ time (the analysis is the same as that for the case (A)).
\end{proof}

\subsection{Computing \texorpdfstring{$\eertree(\trie)$}{eertree(T)}}

\begin{theorem} \label{theo:eertree}
For a given trie $\trie$ with $n$ edges over an integer alphabet $\Sigma = [1..n]$,
the edge-sorted $\eertree(\trie)$ with suffix links can be computed in $O(n)$ time and space.
\end{theorem}

\begin{proof}
Since our algorithm computes $\DPal(\trie)$
in increasing order of the lengths,
we can easily compute the tree-topology of $\eertree(\trie)$
in conjunction with $\DPal(\trie)$ in $O(n)$ time and space.
Recall that the algorithm of Theorem~\ref{theo:stree_linear_integer} builds the edge-sorted suffix tree.
Thus, the edges of the obtained $\eertree(\trie)$ are already sorted.

  Recall that a palindrome $q$ is the longest proper suffix palindrome of
  another palindrome $p$ iff $q$ is the longest proper prefix palindrome of $p$.
  Thus, we can compute the suffix links of the nodes of $\eertree(\trie)$
  by a standard traversal on $\STree(\trie)$ in which
  all the palindromic nodes are explicitly inserted as in Figure~\ref{fig:suffix_tree}.
\end{proof}
 \section{Answering queries to find palindromes in sub-paths}\label{sec:algo_strings}

A trie $\trie$ can be regarded as a compact representation of a set $\mathcal{S}_\trie$ of strings that are the path strings from the leaves to the root of $\trie$.
This section presents how to report
all distinct/maximal palindromes in each string in $\mathcal{S}_\trie$ upon query,
in
output-linear time
with $O(n)$ space, where $n$ is the size of $\trie$.
We remark that the total length of the strings in $\mathcal{S}_\trie$
can be as large as $O(n^2)$, and thus, our $O(n)$-size representation of the palindromes for $\mathcal{S}_\trie$ is space-efficient.
In addition, our algorithms which follow permit us to query
on any sub-path for maximal palindromes,
and on any suffix-path for distinct palindromes.

\subsection{Maximal palindromes in a sub-path}
As for maximal palindromes in any path of a given trie, we obtain the following result:
\begin{theorem}
  After $O(n)$-time preprocessing on the input trie $\trie$,
  given a sub-path $(\mathbf{u}$, $\mathbf{v})$ in $\trie$ and a position $\alpha$ that is either a node or an edge on the path $(\mathbf{u}, \mathbf{v})$, the maximal palindrome centered at $\alpha$ in string $\str(\mathbf{u},\mathbf{v}) \in \Substr(\trie)$ can be computed in $O(1)$ time.
\end{theorem}
\begin{proof}
  In the preprocessing,
  we transform the occurrence representation of maximal palindromes into
  the pair $(\mathbf{s}, \alpha)$ of the starting node $\mathbf{s}$ and its center $\alpha$,
  which is either a node or an edge on the suffix-path from $\mathbf{s}$.
  Such transformation can be done in linear time by traversing $\trie$:
  We maintain an array $\mathcal{N}$ of size $\height(\trie)$ that stores nodes in the suffix path from the current node
  so that we can access any ancestor of given depth in constant time.
  When we visit a node $\mathbf{v}$ and find a maximal palindrome $(\mathbf{v}, k)$,
  we can compute the $k/2$-th shallower locus $\alpha$ by using array $\mathcal{N}$ in constant time.
  Simultaneously, we store the pointer from each $\alpha$ to the starting node $\mathbf{s}_{\alpha}$ of the maximal palindrome centered at $\alpha$.

  Given nodes $\mathbf{u}$, $\mathbf{v}$ and center position $\alpha$ as a query,
  we first compute the lowest common ancestor of $\mathbf{s}_\alpha$ and $\mathbf{u}$. Let $\mathbf{w}$ be the LCA node. 
  If the center position $\alpha$ is a node $\mathbf{z}$,
  then the length of the maximal palindrome centered at $\alpha$ in $\str(\mathbf{u},\mathbf{v})$ is $\min\{2|(\mathbf{w}, \mathbf{z})|,2|(\mathbf{z}, \mathbf{v})|\}$.
  Otherwise, namely, if the center position $\alpha$ is an edge $(\mathbf{x}, \mathbf{y})$,
  then the length of the maximal palindrome centered at $\alpha$ in $\str(\mathbf{u},\mathbf{v})$ is $\min\{2|(\mathbf{w}, \mathbf{x})|,2|(\mathbf{y}, \mathbf{v})|\} + 1$.
\end{proof}

\subsection{Distinct palindromes in a suffix-path}

  In general, a substring palindrome $p$ in a string $S$ may occur more than once in $S$.
  It suffices to report a representative of the occurrences of each palindrome for computing the set of distinct palindromes.
The algorithm that computes $\DPal(S)$ for a string $S$ in~\cite{DBLP:journals/ipl/GroultPR10} reports their leftmost occurrences in $S$.
Also, the algorithm that computes $\DPal(\trie)$ for a trie $\trie$ in~\cite{FunakoshiNIBT19} reports the first occurrence of each palindrome during a depth-first traversal in $\trie$.
Below, we show how to efficiently report $\DPal(w)$ for any suffix-path string $w$ in the trie $\trie$.
\begin{theorem}\label{thm:dpalquery}
  After $O(n)$-time preprocessing on the input trie $\trie$,
  given a node $\mathbf{u}$ in $\trie$,
  all the distinct palindromes in the suffix-path string $\suf(\mathbf{u}) \in \Suffix(\trie)$ can be enumerated in output-linear time.
\end{theorem}

\begin{proof}
An occurrence $(\mathbf{v}, |w'|)$ of string $w'$ is called a \emph{rightmost} occurrence in $\trie$
  if there is no occurrence of $w'$ in $\suf(\mathbf{v})$ other than $(\mathbf{v}, |w'|)$.
As the representative of the occurrences of a palindrome $p$ in $\trie$, we choose a rightmost one.
  Our task is to enumerate rightmost occurrences of palindromes in a given suffix $\suf(\mathbf{u}) \in \Suffix(\trie)$.
  Let $\lpp_\mathbf{v}$ be the longest prefix palindrome starting at a trie node $\mathbf{v}$.
  If $(\mathbf{v}, |p|)$ is a rightmost occurrence of a palindrome $p$, then $p = \lpp_\mathbf{v}$ holds
  (if not, $p$ occurs in $\suf(\mathbf{v})$ as a proper suffix of $\lpp_\mathbf{v}$ since they are palindromes, a contradiction).
  Therefore, once we have computed the set 
  $\mathcal{R} = \{ \mathbf{v} \mid (\mathbf{v}, |\lpp_\mathbf{v}|) \text{ is a rightmost occurrence of } \lpp_\mathbf{v} \text{ in } \trie\}$
  and have \emph{marked} the trie nodes in $\mathcal{R}$,
  the set $\DPal(w)$ for $w =\suf(\mathbf{u})$ can be computed
  by searching all the marked nodes on the suffix-path from $\mathbf{u}$ in $\trie$ when a query node $\mathbf{u}$ is given.
The search can be done in $O(|\DPal(w)|)$ time
  by recursively querying the nearest marked ancestor (NMA) from $\mathbf{u}$ until it reaches the root.
  Note that $O(1)$-time static NMA queries can be preprocessed by an $O(n)$-time standard traversal on $\trie$. 

  In the following, we show how to precompute the set $\mathcal{R}$ in linear time.
First, we compute $\lpp_{\mathbf{v}}$ for each node $\mathbf{v}$.
  This can be done in linear time by traversing on $\STree(\trie)$
  since $\lpp_{\mathbf{v}_i}$ corresponds to the nearest palindromic ancestor node of the leaf $\ell_i$ in $\STree(\trie)$.
Also, we store a stack in each palindromic node in $\STree(\trie)$.
Then, we check whether $\lpp_{\mathbf{v}}$ is the rightmost occurrence for each node $\mathbf{v}$
  by performing a depth-first traversal on $\trie$ as follows:
During the traversal, when descending from node $\mathbf{v}$, we push the node onto the stack of the palindromic node representing $\lpp_\mathbf{v}$.
  If the stack is empty just before the push, 
  then $\mathbf{v} \in \mathcal{R}$ holds since $\lpp_{\mathbf{v}}$ is the rightmost occurrence,
  and thus, we mark the node $\mathbf{v}$.
  When ascending from node $\mathbf{v}$, we pop the node from the stack of the palindromic node representing $\lpp_\mathbf{v}$.
\end{proof}

\paragraph*{Longest palindromes in suffix-paths.}
Let $L(\mathbf{v})$ denote the length of a longest palindrome in $\suf(\mathbf{v})$ in the input trie $\trie$.
  For each non-root node $\mathbf{v}$, $L(\mathbf{v}) = \max\{|\lpp_\mathbf{v}|, L(\parent(\mathbf{v})\}$ holds.
  Thus, after computing the lengths of all $\lpp_{\mathbf{v}}$, we can compute all $L(\mathbf{v})$ in a top-down manner.
  Also, we can associate the length $L(\mathbf{v})$ with the corresponding occurrence while computing them.
  As we proved in Theorem~\ref{thm:dpalquery},
  we can compute $\lpp_{\mathbf{v}}$ for each node $\mathbf{v}$ in linear time.
  The next corollary holds:
  \begin{corollary}
    We can compute a longest palindrome in each suffix-path $\suf(\mathbf{v}) \in \Suffix(\trie)$ in $O(n)$ total time.
  \end{corollary}

 \section{Conclusions and open questions}

This paper proposed the \emph{first} $O(n)$-time algorithm which computes all distinct palindromes and all maximal palindromes in a given trie $\trie$ of size $n$, where the edge labels are drawn from an integer alphabet of size polynomial in $n$.

In the case of a general ordered alphabet of size $\sigma$, one can first sort the edge labels of $\trie$ in $O(n \log \sigma)$ time with $O(n)$ space and replace the edge labels with their lexicographical ranks in range $[1..n]$, so our algorithms work in $O(n \log \sigma)$ time and $O(n)$ space.

It is open whether there exists an $O(n)$-time algorithm for computing distinct/maximal palindromes in a trie
for general ordered alphabets.
To achieve this goal, it is prohibited to construct edge-sorted suffix trees since there is an $\Omega(n \log \sigma)$-time lower bound~\cite{FarachFM00}.

It is known that there can be $\Theta(n^{3/2})$ distinct palindromes in an \emph{unrooted} edge-labeled tree of size $n$~\cite{BrlekLP15,GawrychowskiKRW15}, and all of them can be computed in $O(n^{3/2} \log n)$ time~\cite{abs-2008-13209}. It is open whether there is an $O(n^{3/2})$-time solution for this problem.
 
\vspace*{1pc}
\noindent \textbf{Acknowledgements.}
We would like to thank anonymous referees for their helpful comments.
We also thank Takuya Matsumoto for discussions.

\end{document}